%% file: main.tex
\documentclass[twoside,leqno,twocolumn]{article}

\usepackage[letterpaper]{geometry}
\usepackage{ltexpprt}
\usepackage{algorithm}
\usepackage[noend]{algpseudocode}
\algnewcommand\tTo{\textbf{to}~}
\algnewcommand\tAnd{\textbf{and}~}
\algnewcommand\Or{\textbf{or}~}
\usepackage{amsmath}
\usepackage{standalone}
\usepackage{amssymb}
\usepackage{hyperref}
\usepackage{xcolor,graphicx}
\usepackage{caption}
\usepackage{subcaption}
\usepackage[inline]{enumitem}
\usepackage{wrapfig}
\usepackage[capitalise]{cleveref}
\usepackage{bm}
\usepackage{soul}
\usepackage{multirow}
\usepackage{tikz}
\usetikzlibrary{shapes}
\usetikzlibrary{matrix,positioning,shapes.geometric}

\Crefname{@theorem}{Theorem}{Theorems}

\newcommand{\SSS}{{\cal S}}
\newcommand{\unif}{\mathbb{U}}
\newcommand{\indic}[1]{\mathbb{I}_{#1}}

\newtheorem{definition}{Definition}

\usepackage{listings}
\graphicspath{{./fig/}}
\begin{document}
\newcommand\relatedversion{}

\title{\Large Sequential and Shared-Memory Parallel Algorithms for Partitioned Local Depths}
\author{Aditya Devarakonda\footnote{Department of Computer Science, Wake Forest University.}
\and Grey Ballard\footnotemark[1]
}

\date{}

\maketitle







\begin{abstract} \small\baselineskip=9pt
In this work, we design, analyze, and optimize sequential and shared-memory parallel algorithms for partitioned local depths (PaLD).
Given a set of data points and pairwise distances, PaLD is a method for identifying strength of pairwise relationships based on relative distances, enabling the identification of strong ties within dense and sparse communities even if their sizes and within-community absolute distances vary greatly.
We design two algorithmic variants that perform community structure analysis through triplet comparisons of pairwise distances.
We present theoretical analyses of computation and communication costs and prove that the sequential algorithms are communication optimal, up to constant factors.
We introduce performance optimization strategies that yield sequential speedups of up to $29\times$ over a baseline sequential implementation and parallel speedups of up to $19.4\times$ over optimized sequential implementations using up to $32$ threads on an Intel multicore CPU.
\end{abstract}

\input{text/intro}
\input{text/background}
\input{text/algorithms}

\input{text/analysis}
\input{text/sequential_perf}
\input{text/parallel_perf}
\input{text/application}
\input{text/conclusion}

\bibliographystyle{siam}
\bibliography{refs}

\clearpage
\appendix
\input{text/appendix}
\end{document}

%% file: text/intro.tex
\section{Introduction.}
Partitioned local depths (PaLD) is a method for revealing community structure in distance-based data \cite{pald_pnas22}.
Given pairwise distances (or dissimilarities) of a set of points, PaLD computes another pairwise measure called cohesion that measures closeness based on relative distances.
By relying on relative distance, PaLD is able to use a universal threshold to distinguish between strong and weak ties without defining neighborhoods by a single number of neighborhoods, neighborhood size, or absolute distance threshold.
In this way, PaLD can identify neighborhoods of varying size and density, making it useful for data where the relationships among points behave differently across the space.

The input to PaLD is a distance matrix, and the output is a cohesion matrix.
As detailed in \cref{sec:back}, computing cohesion requires determining the size of the local neighborhood of each pair of points and then computing contributions to cohesion values based on neighborhood sizes.
In each case, the fundamental operation is a comparison of the pairwise distances among triplets of points.
Given $n$ points, this yields an arithmetic complexity of $O(n^3)$.
The goal of this paper is to develop efficient sequential and shared-memory parallel algorithms for scaling PaLD to datasets of size up to $O(10^5)$, making it computationally feasible to analyze ones that fit in memory on a single server.
\Cref{sec:alg-design} presents the structure of the PaLD computation and our two main algorithmic approaches, which we call pairwise and triplet, respectively.
As an $O(n^3)$ computation, PaLD shares many similarities with dense matrix multiplication (GEMM), and our algorithmic design borrows from ideas of cache-efficient algorithms for GEMM \cite{bacd-97-ics,GG_toms08,SLLvdG19}.
For example, the basic computation is a comparison between distances of points $x,y,z$, which involves distance matrix entries $d_{xy}$, $d_{yz}$, and $d_{xz}$ and has an access pattern similar to the fused multiply-adds (FMA) within GEMM.
There are a few key differences between PaLD and GEMM.
First, because of symmetric distances, the order of the points is irrelevant, so rather than requiring consideration of all $n^3$ possible values of $x,y,z$, we need consider only $\binom{n}{3}\approx n^3/6$ unique triplets.
Second, while the memory access of distances is regular, the updates of the cohesion requires branching based on distance comparisons.
Finally, the computation requires two passes because cohesion updates depend on the sizes of local neighborhoods.
Each pass requires a varying mix of integer and floating point operations in addition to the branching.
The pairwise and triplet approaches navigate a tradeoff between exploiting symmetry and achieving regular data access and parallelization.

In \cref{sec:analysis} we prove a lower bound on the cache efficiency of any PaLD algorithm, and we show that both of our algorithms achieve optimal cache performance, up to constant factors.
By exploiting symmetry and applying cache blocking, we obtain data locality in cache and minimize the number of reads and writes of matrix values. 
\Cref{sec:seqexp} details our low-level optimizations of the two PaLD algorithms. We show that branch avoidance has the highest impact on sequential performance given the high cost of branch misprediction \cite{hp-comparch,j88-sigarch,jw89-asplos}.
Along with other optimizations including cache blocking and vectorization, we show performance improvements over naive sequential code of up to $29\times$.
In \cref{sec:par-algs} we design, optimize, and evaluate OpenMP parallel versions of the two PaLD algorithms.
We show that the pairwise algorithm enables regular data access patterns and loop-based parallelism that can largely avoid write conflicts.
The triplet algorithm exploits more symmetry to reduce arithmetic operations but requires task-based parallelism due to more complicated data access patterns and write conflicts.
We also apply Non-Uniform Memory Access (NUMA) optimizations when scaling across sockets. 
We achieve strong scaling speedups up to $19.4\times$ for pairwise and $13.2\times$ for triplet over their optimized sequential versions on $32$ threads.
Finally, we describe a text analysis application in \cref{sec:application}, demonstrating the utility of PaLD on larger datasets than previously considered, and we show a parallel speedup of $16.7\times$ on a task with $n=2712$ using $32$ threads.

%% file: text/background.tex
\section{Background.}
\label{sec:back}
Given a set of points and a pairwise distance metric, partitioned local depth (PaLD) algorithms determine the pairwise cohesion between all pairs of points in a dataset \cite{pald_pnas22}.
Assuming that the dataset comprises sufficiently separated subsets, cohesion values are invariant to contraction and dilation of distances within subset distances.
Community structure revealed by cohesion values capture the concept of near neighbors based on relative positioning, adapting to varying density.
This approach is more flexible than standard cluster labeling or nearest neighbor approaches.
Density-based approaches (e.g. DBSCAN) \cite{cm+13-pakdd,cm+15-tkdd,ek+96-kdd} that attempt to combine points into high- and low-density groups based on pairwise distances include thresholding (tuning) parameters to reflect locality and cluster size.
Likewise, $k$-nearest neighbor (KNN) approaches \cite{gh04-neurips} attempt to group points via comparisons against their $k$ nearest neighbors (using absolute distances).
The tuning parameter, $k$, controls the neighborhood size for a given point and is often fixed for all points.
Cohesion values depend on triplet distance comparisons (as opposed to absolute distances) which require only measures of relative similarity and can be more reliable than exact numerical distances for analyzing high dimensional, non-Euclidean data.
PaLD requires $O(n^3)$ operations to compute cohesion values without assumptions on underlying probability distribution or tuning parameters.

Given a set of points $\SSS$, the \emph{local focus} of a pair of points $x,y\in \SSS$ is the set of all points within distance $d_{xy}$ of either $x$ or $y$, where $d_{xy}$ is the distance between $x$ and $y$: $ {\cal U}_{xy} = \{z \in \SSS \; | \; d_{xz} \leq d_{xy} \; \Or \;  d_{yz} \leq d_{xy}\}.$
We let $u_{xy} = |{\cal U}_{xy}|$ denote the size of the local focus.

The \emph{local depth} of a point $x \in \SSS$ is the probability that, given a uniformly chosen random second point $Y\in \cal S$ and a random third point $Z$ chosen uniformly from the local focus ${\cal U}_{xY}$, $Z$ is closer to $x$ than $Y$:
\begin{equation}
\label{eq:LD}
    \ell_x = \Pr\left[d_{Zx} < d_{ZY} \; | \; Y \sim \unif(\SSS {\setminus} \{x\}), Z \sim \unif({\cal U}_{xY})\right].
\end{equation}

The cohesion of a point $z$ to another point $x$ is a part of the local depth $\ell_x$ and is defined as
\begin{equation}
\label{eq:cohesion}
    c_{xz} = \Pr\left[Z = z \; \tAnd \; d_{Zx} < d_{ZY} \right].
\end{equation}
The random variables $Y$ and $Z$ in \cref{eq:cohesion} are chosen from the same distributions as in \cref{eq:LD}; we drop the notation here and later.
This implies that $\ell_x = \sum_{z\in \SSS} c_{xz}$, or that cohesion is partitioned local depth.
The cohesion matrix, $C$, can be used to analyze community structure.
For example, two points have particularly strong cohesion if the impact of one of the points to the other is more than that expected from a random focus point of another random point.

%% file: text/algorithms.tex
\section{PaLD Algorithms Design.}
\label{sec:alg-design}

In order to compute the cohesion of all pairs of points, we can again use the law of total probability to partition $c_{xz}$ across all points $y\in {\cal S}$:
\begin{equation*}
    c_{xz} = \sum_{y\in \SSS} \Pr\left[Y = y \; \tAnd \; Z = z \; \tAnd \; d_{Zx} < d_{ZY} \right].
\end{equation*}
Using the law of conditional probability, this becomes
\begin{multline*}
     c_{xz} = \sum_{y\in \SSS} \Pr\left[d_{zx} < d_{zy} \; | \; Y = y, Z = z \right] \\
     \cdot \Pr\left[Z=z \; | \; Y = y\right] \cdot \Pr\left[Y=y\right]
\end{multline*}
which implies
\begin{equation}
\label{eq:fxz}
c_{xz} = \sum_{y\in \SSS} \indic{d_{xz} \leq d_{yz}} \cdot \frac{\indic{d_{xz}\leq d_{xy}}}{u_{xy}} \cdot \frac{1}{n-1}
 = \frac{1}{n-1} \sum_{y\in \SSS} g_{xyz},
\end{equation}
where $\indic{}$ is the indicator function and we have defined
\begin{equation}
\label{eq:gxyz}
    g_{xyz} = \indic{d_{xz} \leq d_{yz}} \cdot \indic{d_{xz}\leq d_{xy}} \, / \, u_{xy}.
\end{equation}

The task is then to compute $g_{xyz}$ for all $x,y,z\in\SSS$, a total of $n^3$ values.
However, only approximately 1/3rd of the $g_{xyz}$ values are nonzero because, given three points with unique pairwise distance values, only one pair has the minimum distance.
For example, given points $x,y,z\in\SSS$, if $x$ and $y$ are the closest pair, then $g_{xzy}$ and $g_{yzx}$ are nonzero, but $g_{xyz}=g_{yxz}=g_{zxy}=g_{zyz}=0$.
To compute the nonzero values $g_{xzy}$ and $g_{yzx}$, we need the values $u_{xz}$ and $u_{yz}$.
The size of any given local focus can be computed as $u_{xy} = \sum_{z\in\SSS} \indic{d_{xz} \leq d_{xy} \,\Or d_{yz} \leq d_{xy}}.$

We consider two algorithmic approaches to computing the local focus sizes and the final cohesion matrix that take advantage of the symmetry.
The first approach, which we call the \emph{pairwise algorithm}, considers all $\binom{n}{2}$ pairs of points, and for each pair, first determines the size of its local focus and then computes contributions to the cohesion matrix from all points within the local focus.
The second approach, which we call the \emph{triplet algorithm}, considers all $\binom{n}{3}$ triplets of points, and for each triplet, determines which of the two local foci the triplet contributes to and then (in a second pass) determines which of the two cohesion matrix entries the triplet contributes to.
We analyze and compare the two algorithms in \cref{sec:analysis}.

\subsection{Pairwise Algorithm.}

The entry-wise pairwise algorithm is given as \cref{alg:pairwise}.
The idea is to perform the computations for each pair of points $x$ and $y$.
To compute $g_{xyz}$ for each third point $z$, we first must compute the size of the local focus, $u_{xy}$.
This requires a pass over all $n$ points with two comparisons and a possible integer increment.
A second pass over all $n$ points determines, for points in the local focus, which of the points $x$ or $y$ the third point supports, and the cohesion matrix is updated accordingly.
Note that only one local focus size need be stored at any one time, requiring minimal temporary memory.

\begin{algorithm}[t]
\footnotesize
    \caption{Pairwise Sequential Algorithm}
    \label{alg:pairwise}
    \begin{algorithmic}[1]
    \Require $D \in \mathbb{R}^{n \times n},$~Distance Matrix
    \Ensure $C \in \mathbb{R}^{n \times n},$~Cohesion Matrix
    \For{$x = 1$ \tTo  $n - 1$}
        \For{$y = x + 1$ \tTo $n$}
            \State{$u_{xy} = 0$}
            \For{$z = 1$ \tTo $n$}
                \If{$d_{xz} < d_{xy}$ \Or $d_{yz} < d_{xy}$}
                    \State{$u_{xy} = u_{xy}+1$}
                \EndIf
            \EndFor
        \For{$z = 1$ \tTo $n$}
            \If{$d_{xz} < d_{xy}$ \Or $d_{yz} < d_{xy}$}
                \If{$d_{xz} < d_{yz}$}
                    \State{$c_{xz} = c_{xz} + 1/u_{xy}$}
                \Else
                    \State{$c_{yz} = c_{yz} + 1/u_{xy}$}
                \EndIf
            \EndIf

        \EndFor
        \EndFor
    \EndFor
\end{algorithmic}
\end{algorithm}

To improve the cache locality, we block the algorithm as follows: instead of considering only one pair of points, we consider two sets of points $\mathcal{X}$ and $\mathcal{Y}$ and consider all the pairs $(x,y)\in \mathcal{X} \times \mathcal{Y}$.
In this way, we obtain locality on the distance matrix block $D_{\mathcal{X},\mathcal{Y}}$ and a temporary block of local focus sizes $U_{\mathcal{X},\mathcal{Y}}$.

As in the entry-wise algorithm, the blocked algorithm makes two passes over all $n$ third points.
The first pass computes $U_{\mathcal{X},\mathcal{Y}}$, a local focus size block, and the second pass makes updates to the cohesion matrix.

\begin{figure}[t]
    \centering
    \includegraphics[width=.9\columnwidth]{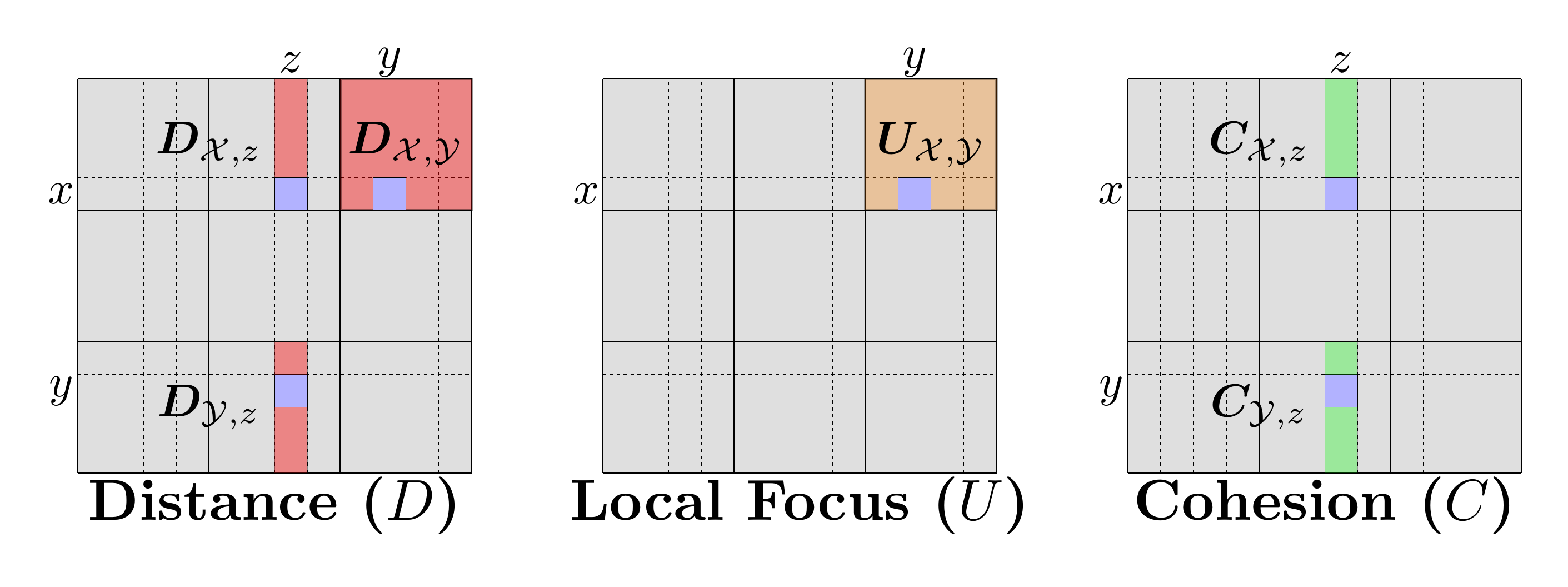}
    \caption{Dependency structure of the blocked pairwise algorithm. The highlighted regions represent quantities with temporal locality. Quantities in red correspond to reads and ones in green correspond to writes. Orange entries are computed and used in fast memory and then discarded.  Blue represents entry-wise dependencies within each matrix/vector.}
    \label{fig:pairwise_dependency}
\end{figure}

\Cref{fig:pairwise_dependency} shows the dependencies among the distance, local focus, and cohesion matrices for the blocked ($b = 4$) pairwise algorithm.
The red blocks correspond to the entries of the distance matrix that are read and re-used while processing the pair of blocks $\mathcal{X}$ and $\mathcal{Y}$ (the pattern is the same in both passes, though $D_{\mathcal{X},\mathcal{Y}}$ remains in fast memory through both passes.
The orange blocks represent entries of the local focus matrix which are computed in fast memory during the first pass and used during the second pass.
The green blocks of the cohesion matrix are re-used during the second pass before being written back to slow memory.
The blue blocks represent dependencies between entries of $D, U,$ and $C$ for one entry-wise iteration.

\subsection{Triplet Algorithm.}

The entry-wise triplet algorithm is given as \cref{alg:triplet}.
In \Cref{alg:pairwise}, if a third point $z$ is in the local focus of $x$ and $y$ and is closer to $x$, then only the support of $z$ for $x$ is recorded in $C$ ($c_{xz}$ is updated).
If $z$ is closer to $x$ in a focus with $y$, then $x$ is closer to $z$ in its focus with $y$.
The idea of the triplet algorithm is to minimize the number of distance comparisons.
By performing all the updates for each triplet of points, we can avoid redundant comparisons.
However, this method requires that the local focus sizes are pre-computed for all pairs of points within the triplet, so it requires more temporary memory.
\begin{algorithm}[t]
\footnotesize
    \caption{Triplet Sequential Algorithm.}
    \label{alg:triplet}
    \begin{algorithmic}[1]
    \Require $D \in \mathbb{R}^{n \times n}$ Distance Matrix
    \Ensure $C \in \mathbb{R}^{n \times n}$ Cohesion Matrix
        \State{Initialize $U = \texttt{triu}(2*\texttt{ones}(n),1)$ }
        \For{$x = 1$ \tTo $n-1$}
            \For{$y = x + 1$ \tTo $n$}
                \For{$z = y + 1$ \tTo $n$}
                    \If{$d_{xy} < d_{xz}$ \tAnd $d_{xy} < d_{yz}$}
                        \State{// $x,y$ is closest pair}
                        \State $u_{xz} = u_{xz} + 1$
                        \State $u_{yz} = u_{yz} + 1$
                    \ElsIf{$ d_{xz} < d_{yz}$}
                        \State{// $x,z$ is closest pair}
                        \State $u_{xy} = u_{xy} + 1$
                        \State $u_{yz} = u_{yz} + 1$
                    \Else
                        \State{// $y,z$ is closest pair}
                        \State $u_{xy} = u_{xy} + 1$
                        \State $u_{xz} = u_{xz} + 1$
                    \EndIf
                \EndFor
            \EndFor
        \EndFor
        \For{$x = 1$ \tTo $n-1$}
            \For{$y = x + 1$ \tTo $n$}
                \For{$z = y + 1$ \tTo $n$}
                    \If{$d_{xy} < d_{xz} $ \tAnd $ d_{xy} < d_{yz}$}
                        \State $c_{xy} = c_{xy} + 1/u_{xz}$
                        \State $c_{yx} = c_{yx} + 1/u_{yz}$
                    \ElsIf{$ d_{xz} < d_{yz}$}
                        \State $c_{xz} = c_{xz} + 1/u_{xy}$
                        \State $c_{zx} = c_{zx} + 1/u_{yz}$
                    \Else
                        \State $c_{yz} = c_{yz} + 1/u_{xy}$
                        \State $c_{zy} = c_{zy} + 1/u_{xz}$
                    \EndIf
                \EndFor
            \EndFor
        \EndFor
\end{algorithmic}
\end{algorithm}

We can also block the triplet algorithm to obtain better cache locality.
Instead of a single triplet of points, we consider three blocks $\mathcal{X},\mathcal{Y},\mathcal{Z}$ and all triplets $(x,y,z)\in \mathcal{X} \times \mathcal{Y} \times \mathcal{Z}$.
We obtain locality on cache blocks of all three matrices: distance, local focus, and cohesion.
Note that a first pass is required to compute the local focus matrix in its entirety, and then blocks of the local focus matrix are read from slow memory during the second pass as needed.

\begin{figure}[t]
    \centering
    \includegraphics[width=.9\columnwidth]{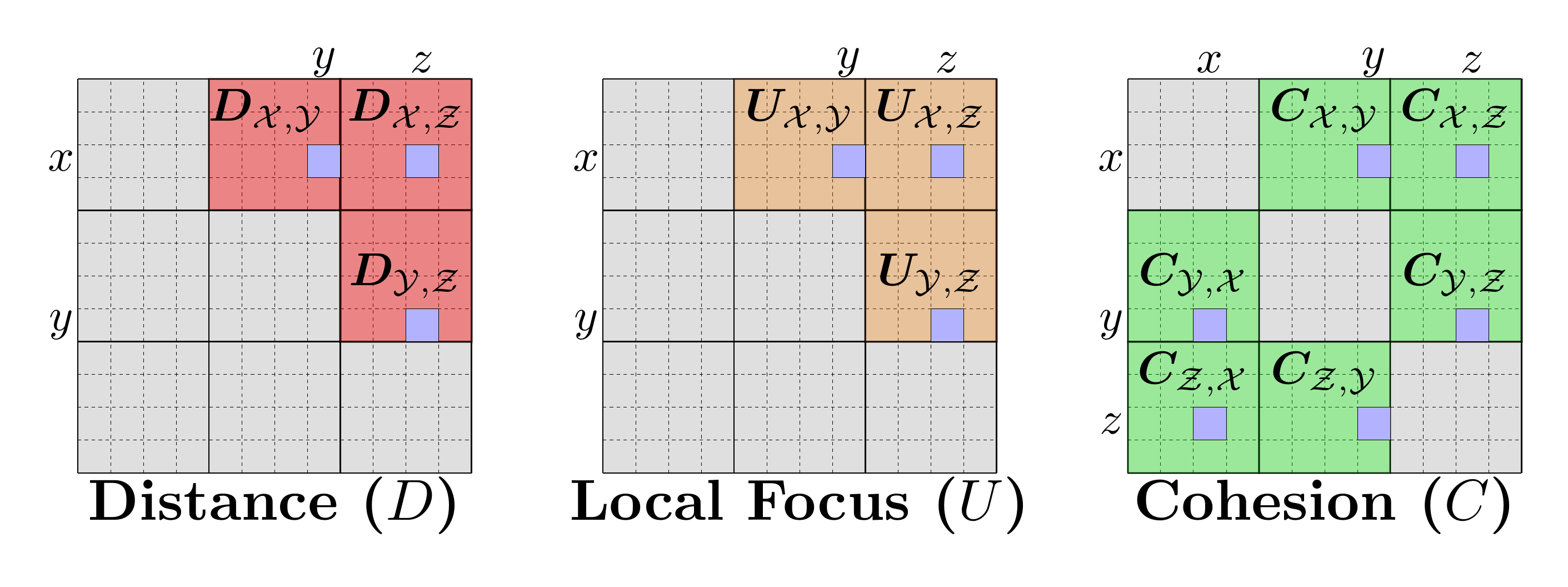}
    \caption{Dependency structure of the blocked triplet algorithm. The highlighted regions represent entries with temporal locality. Matrices in red correspond to reads and ones in green correspond to writes. Matrices in orange correspond to writes during the first pass and reads during the second pass. Blue represents the entry-wise dependencies within each matrix.}
    \label{fig:triplet_dependency}
\end{figure}

\Cref{fig:triplet_dependency} illustrates the dependencies among the distance, local focus, and cohesion matrices for the (blocked) triplet algorithm.
In the first pass, the blocked triplet algorithm reads $3$ blocks from the distance matrix, corresponding to the triplet pairs: $(x,y), (x,z), (y,z)$, and writes to the corresponding $3$ blocks of the local focus matrix.
Note that the distance and local focus matrices are symmetric so only the upper triangular parts are required.
The cohesion matrix is not symmetric, thus in the second pass $6$ blocks must be updated by performing distance comparisons (by reading $D_{\mathcal{X},\mathcal{Y}}, D_{\mathcal{X},\mathcal{Z}}, D_{\mathcal{Y},\mathcal{Z}}$) and utilizing entries of the local focus matrix (by reading $U_{\mathcal{X},\mathcal{Y}}, U_{\mathcal{X},\mathcal{Z}}, U_{\mathcal{Y},\mathcal{Z}}$).

%% file: text/analysis.tex
\section{Sequential Algorithm Analysis.}
\label{sec:analysis}
We model performance using the model, $\gamma F + \beta W$, where $F$ and $W$ represent an algorithm's computation and bandwidth costs, respectively, and $\gamma$ (time per operation) and $\beta$ (time per word moved) represent hardware parameters.
We analyze communication cost assuming a two-level memory hierarchy, which contains fast memory (cache) of size $M$ words and slow memory (DRAM) with unbounded size.
We assume that computation can only be performed on operands residing in fast memory.
If operands are in slow memory, then they must first be read into fast memory.
We limit analysis in this section to a two-level memory hierarchy, but this memory model can be used to analyze communication for each adjacent pair of levels in a multi-level memory hierarchy.
\subsection{Communication Lower Bounds.}

We use the framework in \cite{BCD+_actanumerica14} to derive communication lower bounds.
The lower bound of \cite[Theorem 2.6]{BCD+_actanumerica14} applies to all three-nested-loops (3NL) computations as defined in that paper.
We reproduce the 3NL definition here using the same notation, with sets $S_a, S_b, S_c \subseteq [n] \times [n]$ where $[n] = \{1,2,\dots,n\}$ and mappings $\mathbf{a}:S_a \rightarrow \cal{M}$, $\mathbf{b}:S_b \rightarrow \cal{M}$, $\mathbf{c}:S_c \rightarrow \cal{M}$, where $\cal{M}$ is slow memory.
For each $(i,j)\in S_c$, we also have a set $S_{ij} \subseteq [n]$.

\begin{definition}[{\cite[Definition 2.4]{BCD+_actanumerica14}}]
A computation is considered to be three-nested-loops (3NL) if it includes computing, for all $(i,j)\in S_c$ with $S_{ij}$,
\begin{displaymath}
    \text{Mem}(\mathbf{c}(i,j)) = f_{ij}(\{g_{ijk}(\text{Mem}(\mathbf{a}(i,k)),\text{Mem}(\mathbf{b}(k,j))\}_{k\in S_{ij}}),
\end{displaymath}
where
(a) mappings $\mathbf{a}$, $\mathbf{b}$, $\mathbf{c}$ are all one-to-one into slow memory, and (b) functions $f_{ij}$ and $g_{ijk}$ depend nontrivially on their arguments.
\end{definition}

We first verify that the cohesion matrix computation defined by \cref{eq:fxz,eq:gxyz} is 3NL when the distance matrix is stored explicitly in memory.
To satisfy the first constraint, we define the mappings $\mathbf{a}$, $\mathbf{b}$, and $\mathbf{c}$ as all mapping onto the distance matrix (that is, each mapping is one-to-one but the three mappings are not disjoint).
Here $\mathbf{a}(x,y)$ maps to the distance matrix entry $d_{xy}$.
To satisfy the second constraint, we see that computing $g_{xyz}$ depends nontrivially on $\mathbf{a}(x,y)$ and $\mathbf{b}(y,z)$, as both values must be compared with $d_{xz}$ to evaluate the indicator functions, and computing $c_{xz}$ depends nontrivially on its arguments, as it computes the sum over all values.
As argued in \cref{sec:alg-design}, the number of 3NL operations is $\sum_{i,j} |S_{ij}| = O(n^3)$.
Then, by \cite[Theorem 2.6]{BCD+_actanumerica14}, the bandwidth cost lower bound for PaLD is $W=\Omega(n^3/\sqrt M)$.

\subsection{Cost Analysis.}\label{sec:seqcost}

The blocked algorithms are described in \cref{sec:alg-design} with memory reference patterns depicted in \cref{fig:pairwise_dependency,fig:triplet_dependency}.
The loop structures of the blocked algorithms are shown (with OpenMP parallelization) in \cref{fig:omp_pairwise,fig:omp_triplet}.
We focus on the sequential costs in this section and discuss parallelization in \cref{sec:par-algs}.
Since the algorithms require mixed comparison and arithmetic instructions, we explicitly define the hardware parameters $\gamma_{cmp}$ and $\gamma_{fma}$ to represent the time per instruction for floating-point comparisons and FMAs, respectively.
We ignore the cost of integer arithmetic.
\Cref{fig:omp_pairwise} shows the loop structure of the blocked pairwise algorithm where inner loop computations match \cref{alg:pairwise}.
We use $b$ to represent the block size for the pairwise algorithm.

\begin{theorem}\label{thm:pairwise}
The blocked pairwise algorithm has the leading order computation and communication costs:
    \begin{align*}
        F &= (5\gamma_{cmp} + 1\gamma_{fma})\cdot n\binom{n}{2} \approx 3n^3~\text{flops.}\\
        W &= 4\sqrt{2}~\frac{n^3}{\sqrt{M}} \approx 5.7~\frac{n^3}{\sqrt{M}}~\text{words moved.}
    \end{align*}
\end{theorem}
\begin{proof}

    The blocked pairwise algorithm selects $\binom{n/b+1}{2}$ unique sets of points $\mathcal{X}, \mathcal{Y}$ with $|\mathcal{X}| = |\mathcal{Y}| = b$.
    A total of $nb^2$ iterations are required to determine if a third point, $z$, is in the local focus for each $(x,y) \in \mathcal{X} \times \mathcal{Y}$.
    The local focus update requires $2$ floating-point comparisons followed by $1$ integer accumulate into $u_{xy}$.
    The cohesion update requires $3$ floating-point comparisons and $1$ FMA, as the reciprocals of elements of $U_{\mathcal{X},\mathcal{Y}}$ can be pre-computed once.
    When $\mathcal{X} = \mathcal{Y}$, only $n\binom{b}{2}$ iterations are required to perform local focus and cohesion updates.
    There are $n/b$ such overlapping sets.
    Multiplying over the iterations, summing the work over the local focus and cohesion update loops, and multiplying by $\gamma_{cmp}$ and $\gamma_{fma}$ yields the computation cost.

    Each of the $\binom{n/b + 1}{2}$ possible combinations of $\mathcal{X} \times \mathcal{Y}$ points requires reading the $b\times b$ block $D_{\mathcal{X},\mathcal{Y}}$ from slow memory.
    In the first pass to compute the local focus sizes, for each third point, $z$, we read the two $b \times 1$ vectors $D_{\mathcal{X},z}$ and $D_{\mathcal{Y},z}$ from slow memory.
    The local focus block $U_{\mathcal{X}, \mathcal{Y}}$ is computed and remains resident in fast memory.
    Similarly, each iteration of the second pass cohesion update requires reading the $b \times 1$ vectors $D_{\mathcal{X},z}, D_{\mathcal{Y},z}, C_{\mathcal{X},z}$ and $C_{\mathcal{Y},z}$ from slow memory.
    After each iteration within the second pass, $C_{\mathcal{X},z}$ and $C_{\mathcal{Y},z}$ must be written to slow memory.
    We must maintain $2 b^2$ words of data in fast memory for $D_{\mathcal{X},\mathcal{Y}}$ and $U_{\mathcal{X},\mathcal{Y}}$, along with a constant number of length-$b$ vectors, so $b \leq \sqrt{{M}/2}$ to leading order.
    Multiplying and summing these reads and writes over all iterations yields the leading order communication cost $4 n^3/b$, and choosing $b \approx \sqrt{{M}/2}$ yields the result.
\end{proof}

\Cref{fig:omp_triplet} shows the loop structure of the blocked triplet algorithm, and the inner loop computations match \cref{alg:triplet}.
The local focus sizes and cohesion matrix updates are computed in two separate passes, and two block sizes $\hat b$ and $\tilde b$ can be tuned independently.

\begin{theorem}\label{thm:triplet}
The blocked triplet algorithm has the leading order computation and communication costs:
    \begin{align*}
        F &= (6\gamma_{cmp} + 2\gamma_{fma})\cdot \binom{n}{3} \approx 1.33n^3~\text{flops.}\\
        W &= \left(\sqrt{6} + 4\sqrt{3}\right)\frac{n^3}{\sqrt{M}}\approx 9.4 \frac{n^3}{\sqrt{M}}~\text{words moved.}
    \end{align*}
\end{theorem}
\begin{proof}
    The blocked local focus and cohesion matrix passes have the same loop structure, each selecting $\binom{n/b + 2}{3}$ triplets of sets $\mathcal{X}, \mathcal{Y},$ and $\mathcal{Z}$ each of size $b$ points, though the value of $b$ differs in the two passes.
    The triplet algorithm contains $3$ types of symmetry: $\mathcal{X} = \mathcal{Y} = \mathcal{Z}$, $\mathcal{X} \neq \mathcal{Y} = \mathcal{Z}$, and $\mathcal{X} = \mathcal{Y} \neq \mathcal{Z}$.
    While our implementation accounts for each type of symmetry, we ignore it in our leading order cost analysis.
    The local focus and cohesion update inner iterations each require $3$ distance comparisons to determine the pair of points with minimum distance.
    The cohesion update iteration additionally requires $2$ FMAs to update entries of the cohesion matrix.
    Multiplying operations by their respective $\gamma$ terms and summing work over the two passes proves the computation cost. 

    There are $\binom{n/\hat{b} + 2}{3}$ possible combinations of triplet blocks in the local focus pass. 
    The local focus update must read 2 $\hat{b} \times \hat{b}$ blocks of $D$, read 2 $\hat{b}\times \hat{b}$ blocks of $U$, and write 2 $\hat{b}\times \hat{b}$ blocks of $U$ from/to slow memory.
    Note that the block $D_{\mathcal{X},\mathcal{Y}}$ can be read and the block $U_{\mathcal{X},\mathcal{Y}}$ read and written only $\binom{n/\hat{b} + 1}{2}$ times since they remain fixed while blocks $\mathcal{Z}$ vary in the innermost loop.
    The cohesion update requires reading 2 $\tilde{b} \times \tilde{b}$ blocks of $D$ and $U$, respectively, followed by reading and writing 4 $\tilde{b} \times \tilde{b}$ blocks of $C$.
    The blocks $D_{\mathcal{X},\mathcal{Y}}$ and $U_{\mathcal{X},\mathcal{Y}}$ are read from slow memory and the blocks $C_{\mathcal{X},\mathcal{Y}}$ and $C_{\mathcal{Y},\mathcal{X}}$ can be read and written $\binom{n/\tilde{b} + 1}{2}$ times.
    The total I/O cost is then $n^3/\hat b + 2 n^3/\tilde b$, assuming that all blocks can be stored in fast memory.
    This requires that $\hat b \leq \sqrt{{M}/6}$ and $\tilde b \leq \sqrt{{M}/12}$ to leading order.
    Choosing block sizes at their approximate maximum value yields the communication cost.
\end{proof}
The constants for the communication cost in \Cref{thm:triplet} can be improved by unblocking the innermost loop over $\mathcal{Z}$ for the local focus and cohesion update passes, which allows for a slightly larger block size.
We use this technique for the pairwise algorithm, and it is useful in practice for matrix multiplication as well \cite{SLLvdG19}.
However, incorporating this optimization did not allow for auto-vectorization during cohesion updates where some updates require a stride of $n$.
Blocking all three loops allowed for unit-stride for all cohesion updates. We provide more details in the following section.

We can conclude from \Cref{thm:pairwise,thm:triplet} that the pairwise variant requires more computation than the triplet variant, but it moves less data.
Both sequential variants attain the 3NL lower bound of $\Omega({n^3}/{\sqrt{M}})$ and are communication-optimal within a constant factor.
We will show in the next section how additional performance optimizations can yield large speedups.
The optimized sequential algorithms serve as the baselines from which we derive efficient shared-memory parallel algorithms.

%% file: text/sequential_perf.tex
\section{Sequential Performance Optimization.}\label{sec:seqexp}
We study the performance improvements achieved by each optimization, the tuning parameters introduced, and performance tradeoffs between the pairwise and triplet variants.
All algorithms were written in C and compiled with the Intel C compiler (\verb|icc|) release 2021.06.
The code was compiled with the following compiler flags: \verb|-Ofast -mavx512 -opt-zmm-usage=high|.
Experiments are performed on a single-node, dual-socket platform with two Intel Xeon Gold 6226R CPUs (16 cores per socket).
We run 5 trials for each experiment and use the mean to compute speedups.
We observe low runtime variance across trials, so we omit error bars for simplicity.
We perform experiments on randomly generated distance matrices for powers of two $n \in \{128, \ldots,4096\}$.
Our code can handle arbitrary square matrix sizes, but we limit performance evaluation to powers of two.
\begin{figure}[t]
    \centering
    \includegraphics[width=.8\columnwidth]{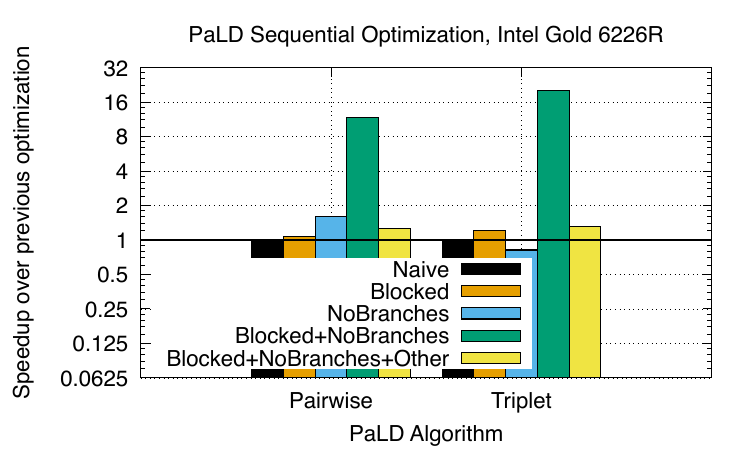}
    \caption{
        Speedup achieved from various performance optimizations applied to the Pairwise and Triplet algorithms.
        Speedups are arranged by algorithm and relative to the previous performance optimization attempted.
        The naive implementations of pairwise and triplet have a speedup of $1$.
    }
    \label{fig:seqopt}
\end{figure}
We begin performance tuning by applying one level of blocking to \cref{alg:pairwise} (naive pairwise) and \cref{alg:triplet} (naive triplet).
We show speedups relative to the previous optimization tried in \cref{fig:seqopt} with a fixed $n = 2048$ matrix.
Overall speedup over naive pairwise (resp. naive triplet) may be obtained by multiplying speedups across all optimizations.
Naive triplet resulted in a speedup of $1.11\times$ over naive pairwise due to less computation.
Introducing one level of blocking to naive pairwise led to a speedup of $1.07\times$.
Applying blocking to the triplet variant led to speedups of $1.20\times$ over naive triplet ($1.33\times$ over naive pairwise).
\Cref{alg:pairwise,alg:triplet} require branches to correctly update $U$ and $C$ based on distance comparisons.
Distance comparisons can be vectorized, but updates to $U$ and $C$ cannot due to branching.
We avoid branches in both algorithms by computing auxiliary mask variables and performing FMAs with these explicit masks.
For \cref{alg:pairwise}, we compute the masks: $r = d_{xz} < d_{xy}~\Or d_{yz} < d_{xy}$ and $s = d_{xz} < d_{yz}$. 
The variable $r$ indicates that $z$ is in the $(x,y)$ local focus and $s$ determines the entry of $C$ to update.
$C$ can be updated via two FMAs: $c_{xz} = c_{xz} + r \cdot s \cdot (1/u_{xy})$ and $c_{yz} = c_{yz} + (r)(1 - s)(1/u_{xy})$.
Branch avoidance introduces a performance tradeoff by increasing computation (e.g. performing FMAs with explicit zeros) but eliminates branch misprediction overhead.
For \cref{alg:pairwise}, branch avoidance enables a fixed stride length for updates of $C$ and facilitates other compiler optimizations (e.g. auto-vectorization and loop unrolling).
Branch avoidance alone yielded a speedup of $1.7\times$ over naive pairwise.
While branch avoidance allows for vectorization, updates to $c_{xz}$ and $c_{yz}$ require a stride length of $n$.
After blocking, we reduce the stride length to $1$ by updating columns of $C$ instead (see \cref{fig:pairwise_dependency}).
The combination achieved speedups of $20.2\times$ over naive pairwise.

\Cref{alg:triplet} must determine the closest pair of points from a triplet $(x,y,z)$.
We avoid branches in \cref{alg:triplet} by computing three masks from three floating point comparisons: $r = d_{xy} < d_{xz}~\tAnd d_{xy} < d_{yz}$, $s = (1 - r)(d_{xz} < d_{yz})$, and $t = (1 - r)(1 - s)$. 
$C$ can then be updated using six FMAs:
\begin{align*}
    c_{xy} = c_{xy} + r\left(1/u_{xz}\right), \quad &
    c_{yx} = c_{yx} + r\left(1/u_{yz}\right),\\
    c_{xz} = c_{xz} + s\left(1/u_{xy}\right), \quad &
    c_{zx} = c_{zx} + s\left(1/u_{yz}\right),\\
    c_{yz} = c_{yz} + t\left(1/u_{xy}\right), \quad &
    c_{zy} = c_{zy} + t\left(1/u_{xz}\right).
\end{align*}
Applying branch avoidance to the triplet algorithm yields a speedup of $0.98\times$ due to the stride-$n$ updates to $C$.
When combined with blocking, however, we attain speedups of $20\times$ over naive triplet.
Triplet with branch avoidance and blocking yields a speedup of $1.1\times$ over pairwise with the same optimizations. 
We were able to extract additional speedup by replacing floating point operations with integer operations during local focus updates, and ignoring equality in pairwise/triplet distance comparisons.
Each entry of $U$ counts the number of points in the local focus based on distance comparisons, with results stored in a mask register.
If $U$ is stored as a floating point array, then each increment to update $U$ requires an expensive integer mask to 32-bit floating point cast operation.
We avoid this by storing $U$ as an integer array during the local focus computation.
This allowed us to combine casting with computing reciprocals prior to cohesion updates.

The theoretical formulation of PaLD \cite{pald_pnas22} allows for ties in pairwise distances (e.g., $d_{xz} == d_{yz}$).
When ties occur, support is split between cohesion entries $c_{xz}$ and $c_{yz}$ (i.e. $c_{xz} = c_{xz} + r\cdot s \cdot \left(0.5/u_{xy}\right)$).
In finite arithmetic, floating point equality is unlikely due to round-off and truncation.
Avoiding ties is critical for \cref{alg:triplet} which contains more distance tie permutations than pairwise.
Introducing these additional optimizations yields self-relative speedups (over naive) of $25.5\times$ and $26.2\times$ for pairwise and triplet, respectively.
Overall, optimized triplet achieves a speedup of $1.14\times$ over optimized pairwise for $n = 2048$.
\begin{figure}[t]
    \centering
    \includegraphics[width=.8\columnwidth]{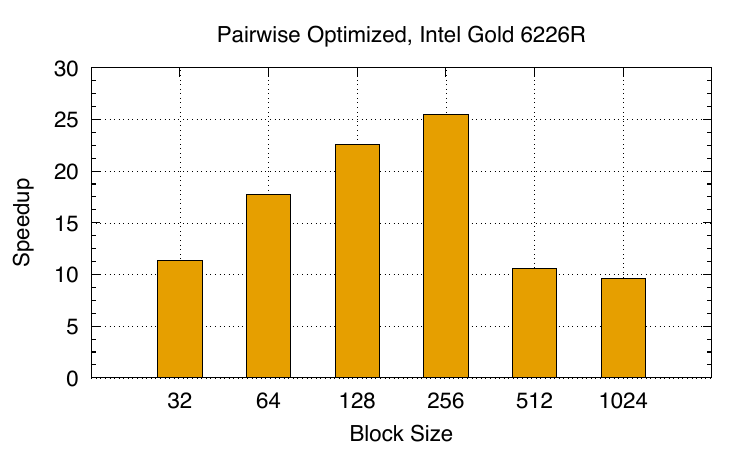}

    \includegraphics[width=.8\columnwidth]{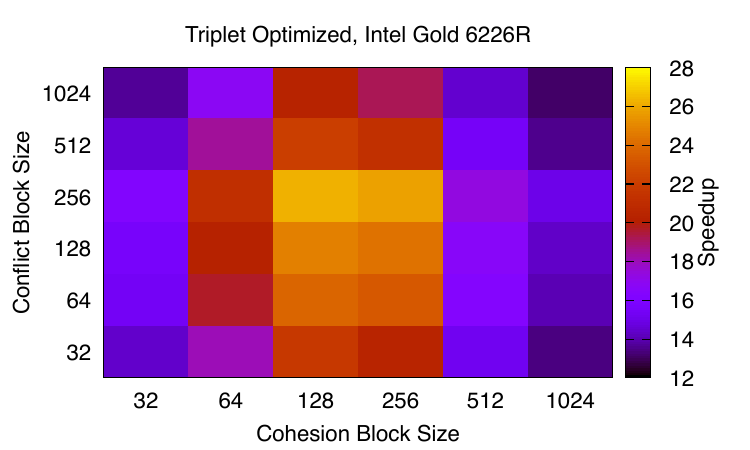}
    \caption{
        Speedup achieved from block size tuning for pairwise (top) and triplet (bottom) for $n = 2048$.
        }
    \label{fig:triplet_heatmap}
\end{figure}
We also perform block size tuning for each algorithm.
We experiment with (powers of two) block sizes in the range $[2^5, 2^{10}]$. Optimized pairwise attains a maximum speedup of $25.5\times$ for $n = 2048$ after tuning.

For optimized triplet, updates to $U$ require storing $3$ distinct blocks of $D$ and $3$ distinct blocks of $U$ in cache.
Updates to $C$ require $3$ distinct blocks of $D$, $3$ distinct blocks of $U$, and $6$ distinct blocks of $C$ in cache.
This suggests that different block sizes may be better than a fixed block size.
\Cref{fig:triplet_heatmap} (bottom) illustrates the speedups observed (over \cref{alg:triplet}) for various block size combinations for the optimized triplet algorithm.
We observe a maximum speedup of $26.2\times$ over naive triplet with $\hat{b} = 256$ and $\tilde{b} = 128$.
\begin{table}[t]
\footnotesize
    \centering
    \begin{tabular}{c|c|c}
        $n$ & Pairwise Optimized & Triplet Optimized\\ \hline\hline
        128 & {\bf 0.00117 (1.58$\times$)} & {0.00185} \\ \hline
        256  & {\bf 0.00497 (1.34$\times$)} & {0.00665} \\ \hline
        512  & {\bf 0.0188 (1.18$\times$)} & {0.0221} \\ \hline
        1024 & 0.1274 & {\bf 0.1208 (1.05$\times$)} \\ \hline
        2048 & 0.9942 & {\bf 0.8734  (1.14$\times$)} \\ \hline
        4096 & 8.3623 & {\bf 6.6111  (1.26$\times$)} \\
    \end{tabular}
    \caption{Running time in seconds (and speedup) comparison of  pairwise and triplet algorithms.}
    \label{tab:seqtimes}
\end{table}
In \cref{tab:seqtimes} we compare running times (and speedups) of optimized pairwise and optimized triplet over a range of input matrix sizes.
For small matrix sizes, where $D,U$ and $C$ all fit in cache, optimized pairwise is fastest (e.g. speedup of $1.58\times$ over triplet at $n = 128$).
This is because $n/b$ is a small integer where lower order terms dominate (see \cref{thm:triplet}).
For larger matrices, optimized triplet performs better (speedup of $1.26\times$ over pairwise at $n = 4096$) due to lower computation cost.
In practice, we expect triplet to be the better sequential variant for most applications of PaLD.
If distances ties must be handled correctly, then pairwise is the better variant due to fewer branches.

Finally, we note that optimized pairwise attains $27.7\%$ of hardware peak at $n = 2048$ and optimized triplet attains $28\%$ at $n = 8192$.
Our Intel CPU has a single-core, single precision peak of $249.6$ Gflops/sec.
Single precision comparisons on our CPU have a cycles-per-instruction (CPI) of $1$ while all other single precision ops have a CPI of $0.5$.
Thus, floating point comparisons are twice as expensive.
See \Cref{sec:pct-peak} for details on percentage of peak calculations for each algorithm.

The combination of all optimizations achieves speedups of $25.5\times$ and $29\times$ for pairwise and triplet, respectively, over naive pairwise (for $n = 2048$).
We observe speedups of $23\times$ and $26.2\times$ over naive triplet.

%% file: text/parallel_perf.tex
\section{Shared-Memory Parallel Algorithms}
\label{sec:par-algs}
This section presents the OpenMP parallelization of the optimized sequential pairwise and triplet algorithms.
\Cref{fig:omp_pairwise} shows the OpenMP version of the blocked pairwise algorithm.
The blocked pairwise algorithm first computes $U_{\mathcal{X},\mathcal{Y}}$ with a pass over all $n$ points $z$.
The local focus $z$-loop can be parallelized across $p$ threads using the OpenMP \verb|parallel for| construct.
All threads must write to $U_{\mathcal{X},\mathcal{Y}}$ so a sum-reduction is required to resolve write conflicts.
The cohesion update pass requires the quantities $1/u_{xy}~\forall~(x,y)~\in~\mathcal{X}{\times}\mathcal{Y}$, which can be parallelized without write conflicts.
Cohesion updates are within each column of $C$ to entries of $C_{\mathcal{X},z}$ and $C_{\mathcal{Y},z}$.
The cohesion pass can be parallelized without write conflicts by splitting the $z$-loop across $p$ threads.
\Cref{fig:pairwise_writes} illustrates the write patterns for optimized OpenMP pairwise for $n = 16$, $b = 4$, and $p = 8$.
Updates to entries of $C$ requires corresponding entries from $D$, so $D$ can also be partitioned column-wise.
The pairwise algorithm is amenable to NUMA optimizations due to the regular data dependencies.

\Cref{fig:omp_triplet} shows the OpenMP version of the blocked triplet algorithm. The triplet approach requires reading all of $D$ for local focus and cohesion update passes.
Blocking is performed over triplets of points, $\mathcal{X},\mathcal{Y},\mathcal{Z}$, and updates to $U$ and $C$ become irregular.
We use the OpenMP tasking model \cite{openmp-spec} for parallelism.
Each triplet block, $\mathcal{X}\times\mathcal{Y}\times\mathcal{Z}$, is a new task that can be executed by any available thread.
Tasks in the local focus pass write to $3$ blocks of $U$. $C$ is not symmetric, so the cohesion update pass writes to $6$ blocks.
Write conflicts arise when multiple tasks need to update the same blocks of $U$ or $C$.
We resolve conflicts by annotating dependencies using the \verb|depend| clause with the \verb|inout| modifier.
\Cref{fig:triplet-tasks} shows the write conflicts for the local focus pass.
Each vertex represents one of the $\binom{n/b + 2}{3}$ tasks and is labeled by $\mathcal{X},\mathcal{Y},\mathcal{Z}$ block values, and edges represent conflicts.
The degree for each vertex varies based on the symmetry in the block.
This leads to irregular dependencies which we will show in \cref{sec:omp-perf} are not as amenable to NUMA optimizations.

\begin{figure}[t]
    \begin{lstlisting}
        for(xb = 0; xb < n/b; ++xb)
          for(yb = 0; yb <= xb; ++yb)
            !\color{orange}{\#pragma omp parallel for \textbackslash}!
            !\color{orange}{reduction(+:U[$\mathcal{X},\mathcal{Y}$])}!
            for(z = 0; z < n; ++z)
              for(x = 0; x < b; ++x)
                y_start = (xb==yb) ? (x+1) : 0;
                for(y = y_start; y < b; ++y)
                  // update !\color{green!75!black}{$u_{xy}$.}!
            !\color{orange}{\#pragma omp parallel for}!
            for(i = 0; i < b*b; ++i)
                U[i] = 1/U[i];
            !\color{orange}{\#pragma omp parallel for}!
            for(z = 0; z < n; ++z)
              for(x = 0; x < b; ++x)
                y_start = (xb==jb) ? (x+1) : 0;
                for(y = y_start; y < b; ++y)
                  // update !\color{green!75!black}{$c_{xz}$ and $c_{yz}$.}!
    \end{lstlisting}
    \caption{Blocked OpenMP pairwise Algorithm.}
    \label{fig:omp_pairwise}
\end{figure}

\begin{figure}[t]
\centering
\includegraphics[width=.8\columnwidth]{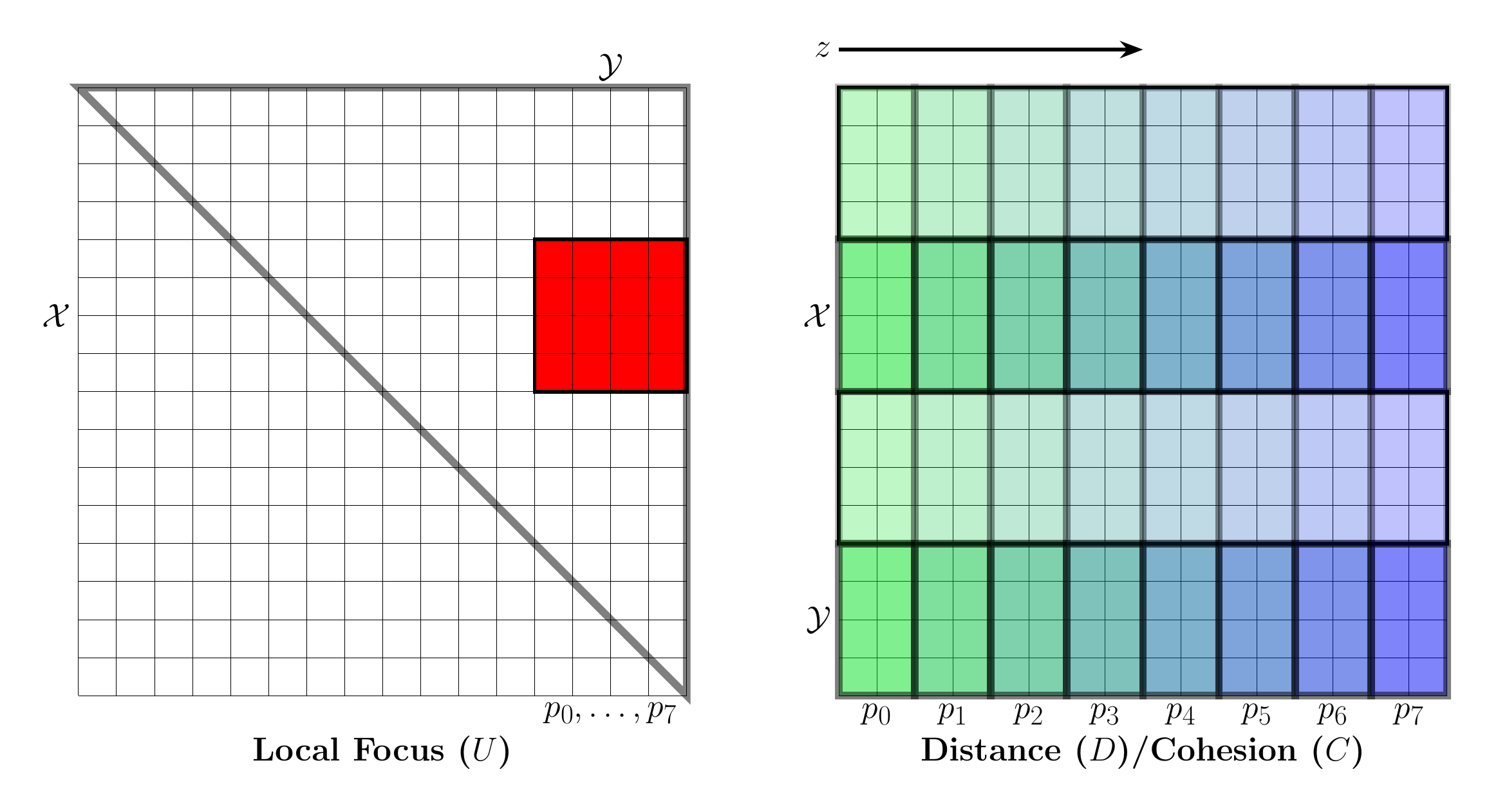}
\caption{Distance matrix reads and Local Focus/Cohesion writes for parallel pairwise code with $n=16$, $b =4$, and $p = 8$.
All threads have write conflicts to the $U$ block for each pair $\mathcal{X},\mathcal{Y}$ (in red), so synchronization is required via reductions.
Only one $U$ block is needed in fast memory at any given point in time. Writes to $C$ are within one column, so column blocks can be partitioned across threads without write conflicts.
}
\label{fig:pairwise_writes}
\end{figure}
\begin{figure}[t]
    \begin{lstlisting}
        !\color{orange}{\#pragma omp single}!
        for(xb = 0; xb < n/b; ++xb)
          for(yb = xb; yb < n/b; ++yb)
            for(zb = yb; xb < n/b; ++zb)
              x_end=(xb==yb && yb==zb)?(b-1):b
              !\color{orange}{\#pragma omp task untied depend(inout, U[$\mathcal{X},\mathcal{Y}$],U[$\mathcal{X},\mathcal{Z}$],U[$\mathcal{X},\mathcal{Z}$])}!
              for(x = 0; x < x_end; ++x)
                y_start=(xb==yb) ? (x+1) : 0;
                for(y = y_start; y < b; ++y)
                  z_start=(yb==zb) ? (y+1) : 0;
                  for(z = z_start; z < zb; ++z)
                    // update !\color{green!75!black}{$u_{xy},~u_{xz},~u_{yz}$.}!
        !\color{orange}{\#pragma omp parallel for}!
        for(i = 0; i < n*n; ++i){
            U[i] = 1/U[i];
        }
        !\color{orange}{\#pragma omp single}!
        for(xb = 0; xb < n/b; ++xb)
          for(yb = xb; yb < n/b; ++yb)
            for(zb = zb; xb < n/b; ++zb)
              x_end=(xb==yb && yb==zb)?(b-1):b
              !\color{orange}{\#pragma omp task untied depend(inout, C[$\mathcal{X},\mathcal{Y}$],C[$\mathcal{X},\mathcal{Z}$],C[$\mathcal{Y},\mathcal{Z}$]) \textbackslash}!
              !\color{orange}{depend(inout, C[$\mathcal{Y},\mathcal{X}$],C[$\mathcal{Z},\mathcal{X}$],C[$\mathcal{Z},\mathcal{Y}$])}!
              for(x = 0; x < xend; ++x)
                y_start=(xb==yb) ? (i+1) : 0;
                for(y = ystart; y < b; ++y)
                  z_start=(yb==zb) ? (y+1) : 0;
                  for(z = z_start; z < zb; ++z)
                    // update !\color{green!75!black}{$c_{xy},~c_{xz},~c_{yz}$,}!
                    // update !\color{green!75!black}{$c_{yx},~c_{zx},~c_{zy}$.}!
    \end{lstlisting}
    \caption{Blocked OpenMP triplet Algorithm.}
    \label{fig:omp_triplet}
\end{figure}
\begin{figure}[t]
\centering
\includegraphics[width=.5\columnwidth]{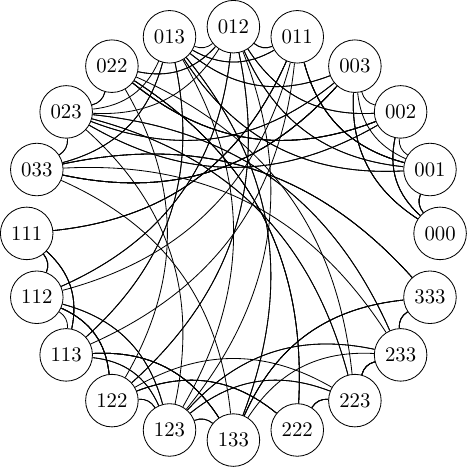}
\caption{Task diagram for parallel triplet with $n/b=4$, where nodes are labeled by their $\mathcal{X},\mathcal{Y},\mathcal{Z}$ block values. Edges represent write conflicts for $U$ between tasks.}
\label{fig:triplet-tasks}
\end{figure}

\subsection{OpenMP Performance.}
\label{sec:omp-perf}
We use OpenMP version 4.5 and test the OpenMP algorithms on randomly generated dense distance matrices with $n \in \{2048, 4096, 8192\}$. We incorporate NUMA optimizations into the pairwise algorithm by controlling thread affinity via the \verb|OMP_PROC_BIND| and \verb|OMP_PLACES| environment variables. We map OpenMP threads to physical cores, by assigning OpenMP thread ids $0$ to $16$ to CPU $0$ and threads $17$ to $31$ to CPU $1$.
\begin{figure}[t]
    \centering
    \includegraphics[width=.75\columnwidth]{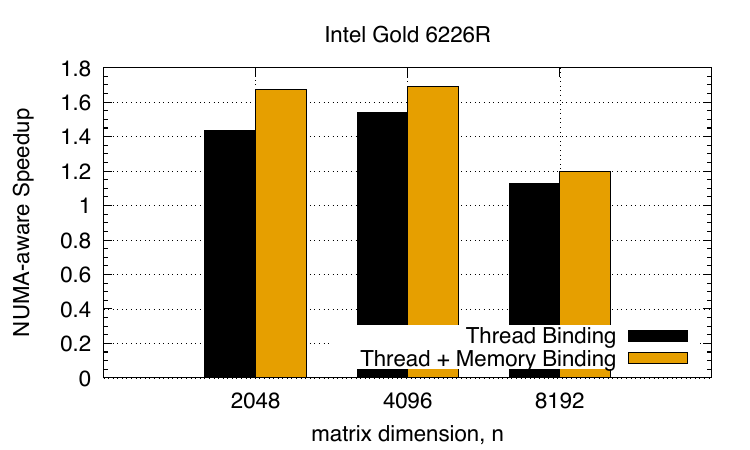}
    \caption{OpenMP pairwise speedup from NUMA optimizations with $n \in \{2048, 4096, 8192\}$ and $p = 32$.}
    \label{fig:numa}
\end{figure}
A static loop schedule yields best performance due to the pairwise algorithm's regular dependencies.
Each thread reads columns of $D$ and $C$ from thread-local fast memory so updates to $C$ are spatially local.
Thread binding ensures that accesses are temporally local by assigning fixed column blocks of $D$/$C$ to threads.
OpenMP allocates memory pages using a first-touch policy by default.
If a single thread allocates $D$, then $D$ resides in the memory hierarchy of the thread's CPU.
$D$ is typically computed outside the scope of the OpenMP algorithms, so we also study the effects of partitioning $D$ across sockets (i.e. memory binding).

\Cref{fig:numa} shows the speedup achieved by introducing thread binding only and thread + memory binding into the OpenMP pairwise algorithm across three matrix sizes, $n \in \{2048, 4096, 8192\}$.
We use the OpenMP pairwise algorithm without NUMA-aware optimizations as our baseline and report speedups for $32$ OpenMP threads.
When we use thread binding only, we observe average speedups of $1.4\times, 1.5\times,$ and $1.13\times$ for $n = 2048, 4098,$ and $8192$, respectively.
Thread binding with memory binding yields average speedups of speedup of $1.7\times, 1.69\times,$ and $1.2\times$ over the baseline.
We did not perform TLB optimizations, therefore, we observe decreasing speedups for large matrix sizes.
We also found that NUMA optimizations are useful at smaller thread counts, $2 \leq p \leq 16$, by mapping half the threads to CPU $0$ and the other half to CPU $1$.
This mapping provides access to the fast memory hierarchies on both CPUs.
We observe speedups ranging from $1.05\times$ ($n = 4096, p = 2$) to $1.33\times$ ($n = 2048, p = 16$) when splitting threads (where $p \leq 16$) across sockets. 
We experimented with thread binding for the OpenMP triplet algorithm but not memory binding due to the irregular data dependencies.
However, we did not observe significant performance improvements over the baseline, so we omit these results from \cref{fig:numa}.
We obtain best OpenMP scaling when using the \verb|untied| clause, which allows suspended tasks to be resumed on any available thread.
Suspended tasks may cause additional reads from slow memory after restart.
Hence, we do not expect NUMA optimizations to be helpful. 
\begin{figure}[t]
    \centering
    \includegraphics[width=.8\columnwidth]{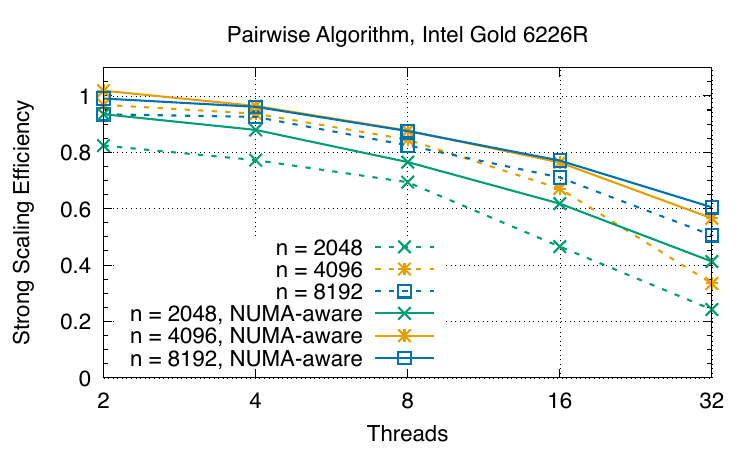}

    \includegraphics[width=.8\columnwidth]{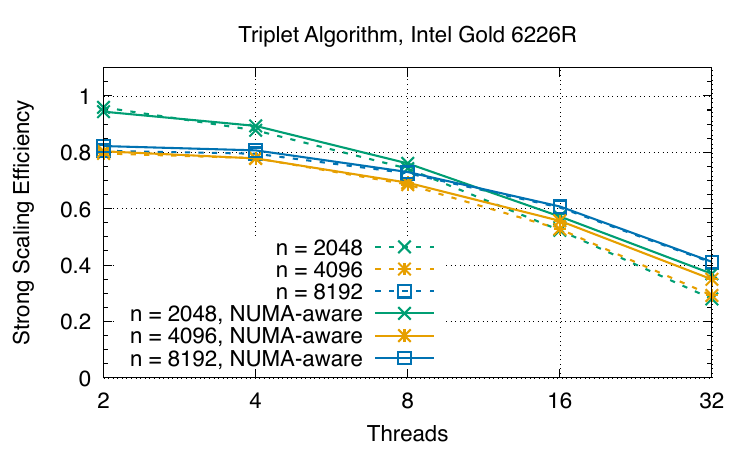}
    \caption{Self-relative strong scaling efficiency of OpenMP Pairwise (top) and Triplet (bottom).}
    \label{fig:strong}
\end{figure} %
We perform strong scaling experiments in \cref{fig:strong} of the OpenMP variants under the same settings as for \cref{fig:numa} and report self-relative efficiency achieved.
We report efficiencies with and without NUMA optimizations. 
The pairwise algorithm without NUMA optimizations achieves efficiencies of $24.2\%, 33.5\%,$ and $50.6\%$ at $p = 32$ for $n = 2048, 4096$ and $8192$, respectively.
Including NUMA optimizations yields efficiencies of $42.9\%, 56.6\%,$ and $60.5\%$ for $p = 32$.
The triplet algorithm achieves efficiencies of $28.0\%, 29.2\%,$ and  $40.9\%$ without NUMA optimizations and $36.9\%, 34.9\%$, and $41.2\%$ with NUMA optimizations for $p = 32$.
The triplet algorithm is the faster sequential baseline, hence the OpenMP triplet efficiencies are lower than those reported for OpenMP pairwise. 
We also study weak scaling of the two algorithms with and without NUMA optimizations.
We fix $n^3/p$ over the range of $p$ tested.
We use the matrix sizes $n_1 \in \{2048, 4096, 8192\}$, where $n_1$ is the matrix size at $p = 1$.
\Cref{fig:weak} shows the results of the weak scaling experiments.
The pairwise algorithm without NUMA optimizations attains weak scaling efficiencies of $30.6\%, 48.2\%$, and $61.4\%$ for $n_1 = 2048, 4096,$ and $8192$, respectively at $32$ threads.
With NUMA optimizations, the efficiencies increase to $59.1\%,~ 63.6\%$, and $65.6\%$ for each of the matrix size settings at $p = 32$.
Triplet without NUMA optimizations achieves weak scaling efficiencies of $44.2\%, 49.1\%,$ and $50.1\%$ and $47.6\%, 49.1\%,$ and $50.1\%$ with NUMA optimizations at $p = 32$. 
\begin{figure}[t]
    \centering
    \includegraphics[width=.8\columnwidth]{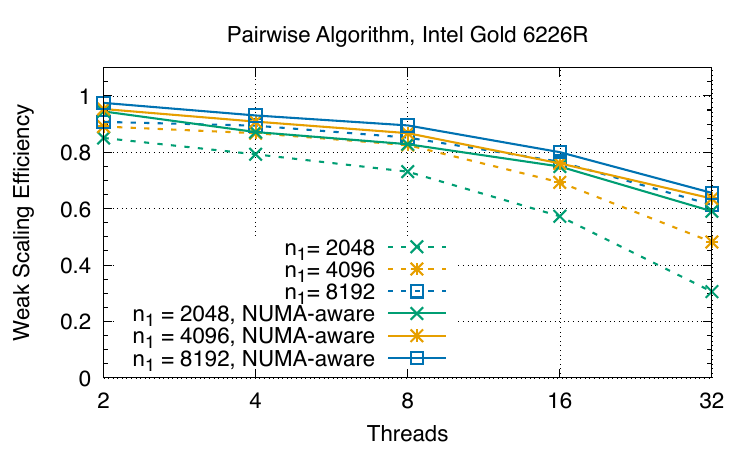}

    \includegraphics[width=.8\columnwidth]{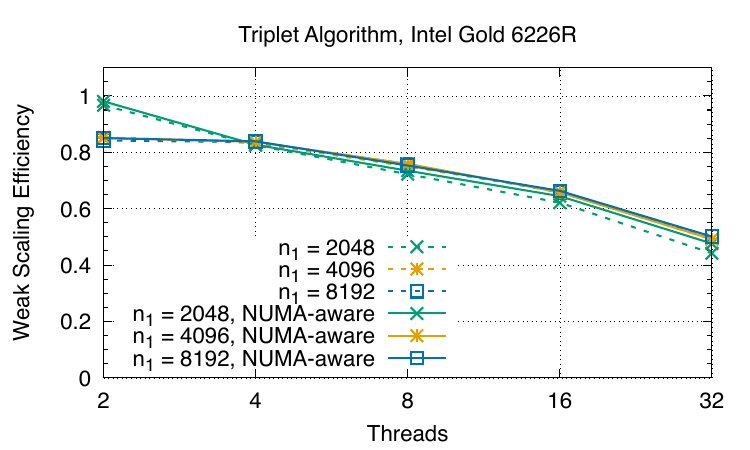}
    \caption{Self-relative weak scaling efficiency of OpenMP Pairwise (top) and Triplet (bottom).}
    \label{fig:weak}
\end{figure}

%% file: text/application.tex
\section{Text Analysis Application}
\label{sec:application}

\begin{figure}
    \centering
    \includegraphics[width=.3\columnwidth,trim=4cm 9cm 4cm 9cm,clip]{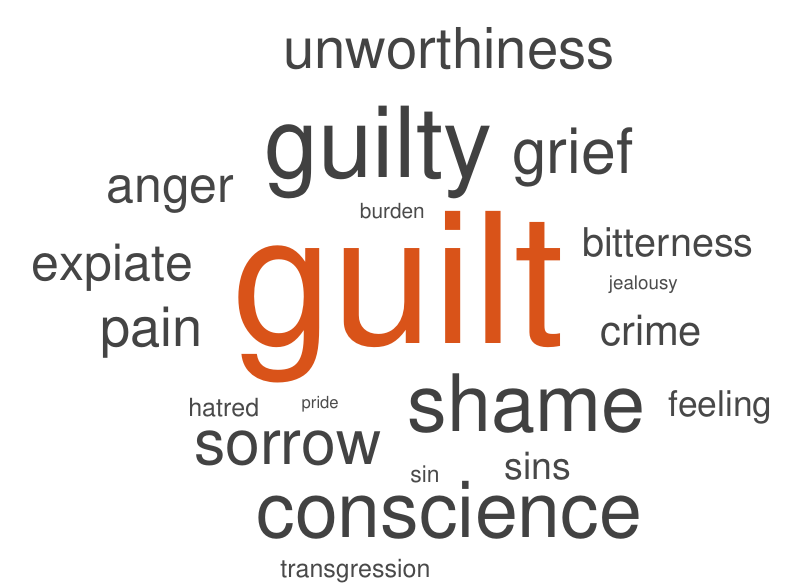}%
    \includegraphics[width=.3\columnwidth,trim=4cm 9cm 4cm 9cm,clip]{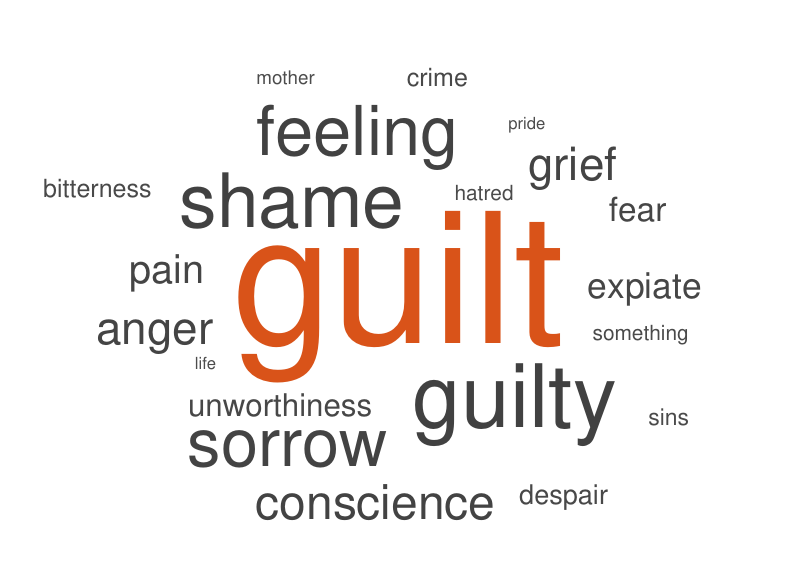}

    \includegraphics[width=.3\columnwidth,trim=5cm 10cm 6cm 10cm,clip]{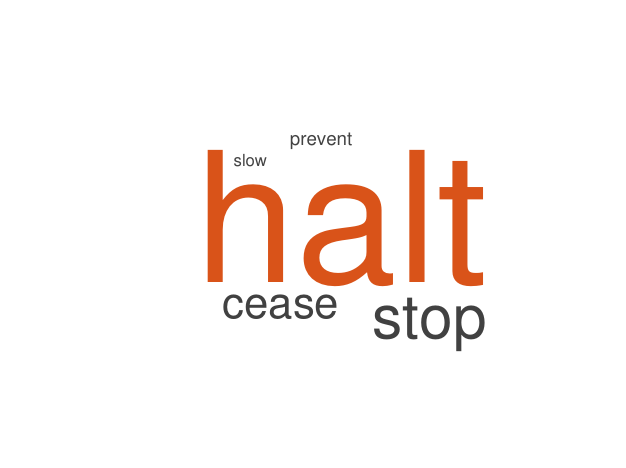}
    \includegraphics[width=.3\columnwidth,trim=5cm 10cm 6cm 10cm,clip]{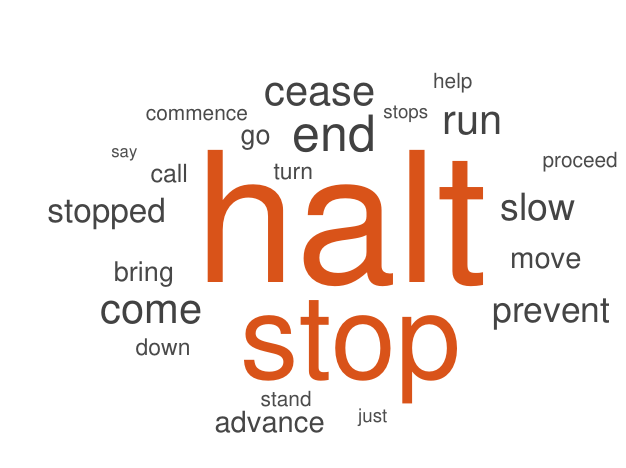}%
    \caption{Word clouds from PaLD analysis (left column) and distance analysis (right column) of the words \emph{guilt} and \emph{halt}. Font size is proportional to cohesion values and inverse distances.}
    \label{fig:wordcloud}
\end{figure}

We demonstrate the utility of PaLD on larger datasets than previously considered \cite{pald_pnas22} for semantic analysis of words extracted from Shakespeare sonnets \cite{matlab-text-toolbox}.
Words are converted to vectors using the pre-trained fastText word embedding \cite{BGJM16-arxiv,JGBM16-arxiv}, yielding a dataset of 2712 words.
We compute Euclidean distance between embedding vectors and generate the cohesion matrix $C$ using the OpenMP pairwise algorithm.
\Cref{fig:wordcloud} shows words associated with \emph{guilt} and \emph{halt} obtained from PaLD and from analyzing only the distance matrix $D$.
PaLD is parameter-free, with strong ties determined by a universal threshold (see \cite{pald_pnas22}), whereas analysis using $D$ requires a user-tuned distance or neighbor-count cutoff.
Note the differing sizes of strong-tie neighborhoods between the two words.
PaLD finds $20$ words with strong ties to \emph{guilt} and 5 words for \emph{halt}.
The $20$ closest words to \emph{guilt} based on distance correspond to a cutoff of $2.26$.
We observe significant overlap between the two sets, though PaLD reports stronger ties to  \emph{expiate} and \emph{conscience}.
PaLD finds $5$ words with strong ties to \emph{halt}.
To illustrate the pitfalls of tuning an absolute distance threshold, we apply the distance cutoff $2.26$ for \emph{halt}, which yields $23$ words including several unrelated ones (e.g. \emph{just} and \emph{say}).
This suggests that absolute distance thresholds are not robust to varying density and distance scales within word neighborhoods.
A distance cutoff of $2.14$ is required for \emph{halt} to match results obtained from PaLD.
Applying the cutoff to \emph{guilt} identifies only $8$ related words, missing several words like \emph{expiate}.
We attain a speedup of $16.7\times$ using the NUMA optimized OpenMP pairwise algorithm at $p = 32$ and an overall run time of $0.178$ seconds.

%% file: text/conclusion.tex
\section{Conclusion}
This paper presents several sequential and shared-memory parallel algorithms for PaLD \cite{pald_pnas22}.
We prove that sequential variants are communication-optimal, up to constant factors.
We illustrate that branch avoidance is critical to attaining high performance; achieving a speedup of up to $29\times$ over naive sequential variants.
Based on our theoretical and empirical studies, we conclude that the triplet variant is the faster sequential algorithm for large matrices due to less computation.
However, we show that the pairwise algorithm is more amenable to parallelization due to regular data dependencies and load balance.
We observe strong scaling speedups  up to $19.4\times$ ($60.5\%$ efficiency), and weak scaling efficiencies of up to $65.6\%$ at $p = 32$ after incorporating NUMA-aware optimizations.
 With the performance achieved on the text analysis application, we show that PaLD can be scaled to nearly any dataset with a distance matrix that fits in the memory of a single server.

\paragraph{Acknowledgements.} We would like to thank Kenneth S. Berenhaut for helpful feedback on the presentation of PaLD and discussions on applying PaLD to semantic analysis of word embedding in \Cref{sec:application}. We would also like to thank Yixin Zhang for code contributions to preliminary versions of the pairwise algorithms. This work is supported by the National Science Foundation under Grant No. OAC-2106920 and the U.S. Department of Energy, Office of Science, Advanced Scientific Computing Research program under Award Number DE-SC-0023296.

%% file: text/appendix.tex
\section{Percentage of Hardware Peak}\label{sec:pct-peak}
\subsection{Pairwise Algorithm.}
This section details the operations counting for the optimized sequential pairwise and triplet algorithms used for percentage of peak calculations in \Cref{sec:seqexp}.
Unless otherwise specified, all operands are assumed to be in $32$-bit floating point format.
The optimized sequential pairwise algorithm requires $2$ comparisons during local focus update to determine if a point, $z$, is in the neighborhood.
The local focus matrix is incremented based on these comparisons.
However, since $U$ is stored in integer format, we ignore the cost of integer increments during the local focus pass.
The cohesion update requires $3$ comparisons: $2$ comparisons to compute mask $r$ which determines if a point $z$ is in the local focus and $1$ comparison to compute mask $s$ which determines the column entry of $C$ to update.
Since results of floating point comparisons are stored in unsigned (integer) format, $r$ and $s$ must be cast to $32$-bit floats before FMAs.
This requires $2$ unsigned int to floating point cast operations.
Finally, $2$ FMAs (each FMA requires two instructions) can be used to update $c_{xz}$ and $c_{yz}$.
We explicitly compute both entries and accumulate with an explicit zero, as this avoids branching.
Based on these operations, the total number of operations for sequential pairwise can be computed as follows:
\begin{align*}
    F = \left(5 \gamma_{cmp} + 2\cdot 2\gamma_{fma} + 2\gamma_{cast}\right)\cdot n \binom{n}{2}
\end{align*}
On our Intel Xeon Gold 6226R CPU, floating point comparisons have a CPI of $1$ whereas FMAs and casting each have a CPI of $0.5$.
Since comparisons are twice as expensive, we normalize our operation count to be relative to FMA/cast.
After normalization, the total number of operations become:
\begin{align*}
    F = 16\gamma\cdot n \binom{n}{2} \approx 8 n^3~\text{ops}.
\end{align*}
Finally, percentage of peak can be calculated by:
\begin{align}\label{eq:pct-peak}
    \frac{1}{249.6}\cdot \frac{F}{10^9 \cdot t_n}
\end{align}
where $t_n$ is the runtime time (in seconds) obtained empirically from executing the optimized sequential pairwise algorithm on a matrix of size $n$ and $249.6$ Gflops/sec is the single precision, single core machine peak of our Intel CPU.
The setting $n = 2048$ and $t_n = 0.99422~\text{seconds}$ (averaged over $5$ trials) yields $27.7\%$ of peak, as reported in \Cref{sec:seqexp}.
\subsection{Triplet Algorithm.}
The optimized sequential triplet algorithm makes two passes: one to compute $U$ in its entirety and one to compute $C$.
The triplet algorithm, which ignores ties, requires $6$ comparisons across the two passes to uniquely determine the pair of points in a triplet with minimum pairwise distance.
The local focus pass and cohesion pass must compute these distances.
Once again, we ignore integer increments in the local focus pass.
The remaining instructions are $3$ casting and $6$ FMA operations to update entries of $C$.
\begin{align*}
        F = \left(12 \gamma_{cmp} + 2\cdot 6\gamma_{fma} + 3\gamma_{cast}\right)\cdot& \binom{n}{3} \approx 6.5 n^3.
\end{align*}
Setting $F = 6.5n^3, n = 8192$ and $t_n = 51.15952$ seconds in \eqref{eq:pct-peak} yields $28\%$ of peak, as reported in \Cref{sec:seqexp}.

\section{Runtime Breakdown}
 \begin{figure}[t]
    \centering
    \includegraphics[width=.8\columnwidth]{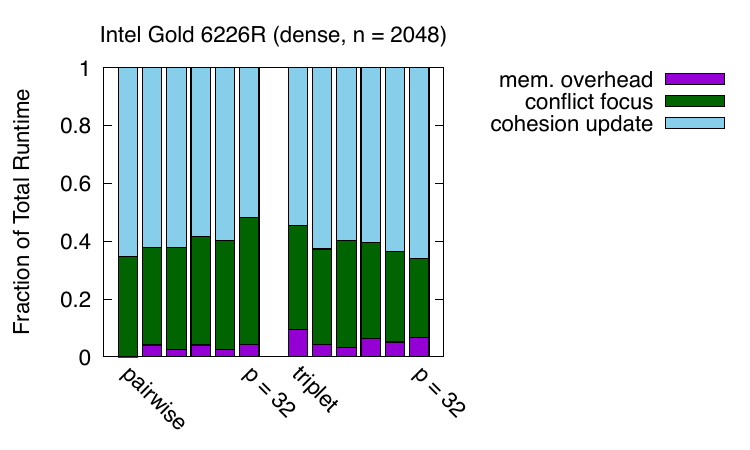}
    \caption{Running time breakdown (as a fraction of total running time) of pairwise and triplet algorithms for a random dense, $n = 2048$ matrix.}
    \label{fig:rt}
\end{figure}

\Cref{fig:rt} shows the running time breakdown of the pairwise and triplet algorithms, grouped by the algorithm.
We report the fraction of time taken to compute the local focus, cohesion update, and memory overhead (e.g. memcpy into explicit cache blocks).
The OpenMP pairwise algorithm requires a reduction during the local focus computation.
As $p$ increases, we see that the local focus computation becomes a barrier to scalability.
The OpenMP triplet algorithm however does not require explicit synchronization.
This results in better scalability for the local focus computation as $p$ increases.
The cohesion update, however, does not scale efficiently since it updates up to $6$ unique blocks of $C$.
Given the irregular task dependencies (see \cref{fig:triplet-tasks}), the OpenMP triplet algorithm was not amenable to NUMA optimizations.
In addition, variance in task costs due to symmetries and a non-static task schedule were barriers to efficiently scaling the cohesion update pass.
Memory overhead is a negligible fraction of runtime for pairwise and triplet algorithms and does not increase with $p$.
At $32$ threads, the pairwise algorithm is faster than triplet algorithm for cohesion update.
The reverse is true for the local focus update.
This behavior may indicate that the two algorithms can be combined by utilizing the triplet approach for local focus update and the pairwise approach for cohesion update for additional speedup.

\section{Scaling on SNAP Datasets.}
\begin{table}[t]
    \centering
    \begin{tabular}{c|c|c|c}
        Dataset & $n$ & sequential & $p = 32$\\\hline\hline
        ca-GrQc & 5242 & 21.69 &\bf 1.390 (15.6$\times$)\\\hline
        ca-HepPh & 12008 & 259.9 & \bf 13.16 (19.7$\times$)\\\hline
        ca-CondMat & 23133 & 1913 &\bf 91.89 (20.8$\times$)\\\hline
    \end{tabular}
    \caption{Pairwise runtimes (in sec.) and maximum speedup over pairwise sequential on SNAP datasets.}
    \label{tab:snapscaling}
\end{table}
We perform scaling experiments on large datasets obtained from the SNAP data repository \cite{snapnets} to illustrate PaLD scalability on collaboration networks.
We obtain distance matrices by computing all-pairs shortest path distances.
\Cref{tab:snapscaling} reports the running times (in seconds) and speedup achieved at $p = 32$ for the pairwise algorithm.
We use the optimized sequential pairwise algorithm as our baseline.
We achieve speedups of $15.6\times$, $19.7\times$ and $20.8\times$ on the ca-GrQC, ca-HepPh, and ca-CondMat datasets \cite{lkf07-tkdd}, respectively.
For the largest dataset, ca-CondMat, we are able to reduce the running time of computing $C$ from $31$ minutes (optimized pairwise sequential) to $92$ seconds (OpenMP pairwise with $p = 32$).